\documentclass{eptcs}
\usepackage[english]{babel}
\usepackage{xspace}
\usepackage{amsmath,amssymb,amsthm}
\usepackage[ruled, vlined]{algorithm2e}
\DontPrintSemicolon
\usepackage{graphicx}
\usepackage{array,caption}
\usepackage{stmaryrd}

\newcommand{\N}{\mathbb{N}}
\newcommand{\Z}{\mathbb{Z}}

\let\code\textsc
\newcommand{\duration}{\code{Duration}}
\newcommand{\goto}{\code{Goto}}
\newcommand{\state}{\code{State}}
\newcommand{\codechoose}{\code{Choose}}
\newcommand{\finish}{\code{Finish}}
\newcommand{\cancel}{\code{Cancel}}
\newcommand{\erase}{\code{Erase}}
\newcommand{\remove}{\code{Remove}}

\newcommand{\attr}{\text{Attr}^*}
\newcommand{\attrone}{\text{Attr}}
\newcommand{\attrtwo}{\text{Attr}^2}
\newcommand{\restrattr}{\text{RestrAttr}}
\newcommand{\cpre}{\text{Pre}}

\newcommand{\circpl}{{\raisebox{1pt}{\scalebox{0.9}{\ensuremath{\bigcirc}}}}}
\newcommand{\boxpl}{\raisebox{-0.5pt}{\scalebox{1.15}{\ensuremath{\Box}}}}

\let\kw\textit

\newcommand{\takeout}[1]{}

\title{The Complexity of Robot Games on the Integer Line
\thanks{This work was partially supported by LIA Informel, the Indo-French Formal Methods lab.}
}

\author{Arjun Arul
  \email{arjun@cmi.ac.in}
  \institute{Chennai Mathematical Institute, India}
  \and Julien Reichert
  \email{reichert@lsv.ens-cachan.fr}
  \institute{LSV, ENS Cachan, France}
}

\def\titlerunning{Robot Games on the Integer Line}
\def\authorrunning{A. Arul \& J. Reichert}

\date{\today}

\usepackage{pgf}
\usepackage{tikz}
\usetikzlibrary{arrows,automata}
\usepackage{verbatim}

\newtheoremstyle{perso}
{5mm} 
{3pt} 
{\itshape} 
{} 
{\bfseries} 
{:} 
{.5em} 
{} 

\theoremstyle{perso}
\newtheorem{theorem}{Theorem}
\newtheorem{corollary}[theorem]{Corollary}
\newtheorem{proposition}[theorem]{Proposition}
\newtheorem{lemma}[theorem]{Lemma}

\begin{document}
\maketitle

\begin{abstract}
In robot games on $\Z$, two players add integers to a counter. 
Each player has a finite set from which he picks the integer to add, and the objective of the first player 
is to let the counter reach $0$. 
We present an exponential-time algorithm for deciding the winner of a robot game given the initial counter value, 
and prove a matching lower bound. 
\end{abstract}

\section{Introduction}
Robot games \cite{DR13} are played by two players, a reacher and an opponent, by updating a vector of $m$ integer counters. 
Each player controls a finite set of integer vectors in $\Z^m$. 
Plays start with a given initial vector $v_0 \in \Z^m$ of counter values, and proceed in rounds.
In each round, first the opponent and then the reacher 
adds a vector from his set to the counter values.
The reacher wins when, after his turn, the vector of counter values is zero.

We consider the problem of determining the winner of a robot game for dimension $m=1$. 
Towards this, we present an algorithm for solving this problem in EXPTIME and show that the bound is hard. 

Robot games are a particular kind of reachability games. Such games are
played on a graph $(Q,E)$, called an arena, where the set of vertices $Q$ is partitioned
into $Q_1$ and $Q_2$ to designate which player is in turn to move.
Here, a play is a (possibly infinite) sequence of vertices $q_0 q_1 \dots$ starting with a given initial vertex $q_0$.
At any stage $i$, if $q_i \in Q_1$ the reacher chooses a successor $q_{i+1}$ of $q_i$ such that $(q_i,q_{i+1}) \in E$; otherwise 
the opponent chooses the successor.
The objective is given by a subset $Q'$ of $Q$: the reacher wins a play if it visits a vertex in $Q'$.

The winning set in reachability games, i.e., the configurations from which the reacher has a winning strategy,
can be computed by the attractor construction \cite{Tho95}.
However, in robot games, we have infinitely many configurations, so we will need further tools. 
It turns out that here, the winning set is closed under linear combinations. 
For dimension one, this implies that there exists a bound such that the winning set becomes easy to describe from this bound onwards.
The key idea of our algorithm is to perform the attractor construction up to 
such a bound, which we compute using a theorem from \cite{Wil78}.

In view of the simplicity of their description, it may come as a surprise that robot games are EXPTIME-hard.
We prove this by reduction from countdown games \cite{JLS07},
another class of reachability games, with a nonnegative counter that can only decrease.

Robot games belong to the family of reachability games on counter systems. 
Such games are played on a labelled graph $(Q,E)$ where the set of edges is $E \subseteq Q \times \Z^m \times Q$
and there is a vector of $m$ counters. When an edge $(q,v,q')$ is taken, the vector of counters is updated by adding $v$ to it.
In counter reachability games, the objective of the reacher is either a set of vectors or a set of pairs (vertex,vector).
We can view robot games as counter reachability games on an arena with only two vertices.

\takeout{
Many counter reachability games in dimension two are undecidable, e.g., on vector addition systems with states,
which disable the edges that would make a counter become negative,
it is undecidable whether the first player has a winning 
strategy when the objective is $C \times ((\{0\} \times \N) \cup (\N \times \{0\}))$,
where $C$ is a subset of the set of vertices \cite{BJK10}.
Nevertheless, the decision problem associated to a robot game and an initial vector of counter values is open for dimension two.
}

We hope that settling the complexity of one-dimensional robot games 
will help improving the EXPSPACE upper bound that follows from \cite{BJK10},
for deciding the winner of counter reachability games on the integer line.

\section{Definitions}

When we write ``positive'' or ``negative'', we always mean ``strictly positive'' or ``strictly negative''.
We write $-\N$ for the set of nonpositive integers.

A \kw{robot game} \cite{DR13} in dimension one is a pair $(U,V)$, where $U$ and $V$ are finite subsets of $\Z$.
The robot game is played by a \kw{reacher}, who owns the subset $U$, and an \kw{opponent}, who owns the subset $V$.
Given an initial counter value $x_0 \in \Z$, a \kw{play} proceeds in \kw{rounds}.
In a round that starts at the counter value $x \in \Z$, the opponent chooses what we call a \kw{move}
$v \in V$ and updates the counter to $x + v$, then the reacher chooses a move
$u \in U$ and updates the counter to $x + v + u$, in which the round ends.
The play ends and the reacher wins it if the round ends at $0$, else a new round is played.
By convention, the reacher wins immediately when a play starts at $0$.

To represent robot games, we draw their two vertices, $\circpl$ for the reacher and $\boxpl$ for the opponent,
and two edges that list the set of each player.

\begin{figure}[h]
\begin{center}
\begin{tikzpicture}[->,>=stealth',shorten >=1pt,auto,node distance=3cm,
                    semithick]
  \tikzstyle{state}=[circle,minimum size=8mm,fill=white,draw=black,text=black]
  \tikzstyle{oppstate}=[minimum size=8mm,fill=white,draw=black,text=black]

  \node[oppstate] (A)                    {};
  \node[state]         (B) [right of=A] {};

  \path (A) edge [bend left] node {$-1,3$} (B)
        (B) edge [bend left] node {$-1,0,4$} (A);

\end{tikzpicture}
\caption{Example of a robot game for the sets $U = \{-1,0,4\}$ and $V = \{-1,3\}$.}
\end{center}
\end{figure}

Formally, a play is a finite or infinite sequence $\Z(VU)^*$ or $\Z(VU)^\omega$.
A play prefix in a robot game is a word $\pi \in \Z(VU)^* \cup \Z(VU)^*V$,
the first letter of this word is the initial counter value and the other ones are the moves players do in the play prefix.
We associate to a play prefix $c_0 v_0 u_0 \dots v_h$ its destination $c_0 + v_0 + u_0 + \dots + v_h$.

A \kw{strategy} for the reacher (resp. for the opponent) is a function $\sigma : \Z(VU)^*V \to U$
(resp. $\sigma : \Z(VU)^* \to V$).
A strategy $\sigma$ is \kw{memoryless} if all play prefixes with the same destination
have the same image under $\sigma$.
We then take the destination of a play prefix $\pi$ instead of $\pi$ itself as argument of a memoryless strategy,
which we define from now on as a function $\Z \to U$ or $\Z \to V$ depending on the player.

A counter value $x$ is \kw{winning} if there exists a reacher strategy,
such that for all strategies of the opponent,
the reacher wins the play that starts at $x$ and in which each player moves according to his strategy.
We switch reacher and opponent in the last sentence to define the notion of a \kw{losing} counter value.
The decision problem associated to a robot game $(U,V)$ and an initial counter value $x \in \Z$
asks whether $x$ is winning.

By the Gale-Stewart theorem \cite{GS53}, robot games are determined:
In any robot game, every initial counter value is either winning or losing.
Robot games are even positionally determined, because they are reachability games,
which means that if a player has a winning strategy, then he also has a memoryless winning strategy.

A \kw{linear set} in $\Z$ is a set of the form $\{x + \sum_{i=1}^n k_i x_i\ |\ k_1,\dots,k_n \in \N\}$, for some integers $x,x_1,\dots,x_n$.
In other words, it is the least set that contains $x$ and is closed under addition of integers in $\{x_1,\dots,x_n\}$.
We denote that set by $x\ + \langle\{x_1,\dots,x_n\}\rangle_\N$ or simply $x + x_1\N$ when $n = 1$.
We also write $\langle Y\rangle_\N$ rather than $0 + \langle Y\rangle_\N$.
We say that an integer is $Y$\kw{-reachable} if, and only if, it belongs to $\langle Y\rangle_\N$.

The \kw{amplitude} of a robot game $(U,V)$ is the integer interval
bounded by the extremal combinations of moves in a round.
We denote it by Ampl$(U,V) = \llbracket\min(V)+\min(U),\max(V)+\max(U)\rrbracket$.
We also define for any $k \in \N$ the integer interval
Ampl$^k(U,V) = \llbracket \min(V)+\min(U)-k,\max(V)+\max(U)+k\rrbracket$.

We now give some basic properties of robot games.
Let us first remark that robot games are invariant under translation: Whenever a player can make a move from $x$ to $x'$,
the same move leads from $y$ to $y - x + x'$.

\label{prop1}\begin{proposition}
If two counter values are winning in a robot game, then their sum is also winning.
\end{proposition}

\begin{proof}
Let $x \in \Z$ and $y \in \Z$ be two winning counter values.
Let $\sigma_x$ and $\sigma_y$ be winning strategies of the reacher from $x$ and $y$.
Because the game is invariant by translation, the reacher can enforce a play
that starts at $x+y$ to visit $y$ after one of his turns with the strategy $z \mapsto \sigma_x(z-y)$.
After this first visit to $y$, the reacher wins by using $\sigma_y$.
He always knows whether $y$ was visited during a play prefix $(x+y) v_0 u_0 \dots v_h$,
a necessary and sufficient condition is that a partial sum $(x+y)+v_0+u_0+\dots+v_i$
is $y$.
\end{proof}

As a consequence, if all counter values in a set $X \subseteq \Z$ are winning in a robot game,
then every $X$-reachable counter value is winning.
This guarantees that the winning set is linear.

The next proposition states what happens when a player can force the counter value to increase or decrease unboundedly.

\label{prop2}\begin{proposition}
Let $(U,V)$ be a robot game.
\begin{itemize}
\item If $\max(V) \ge -\min(U)$, then each positive counter value is losing.
Similarly, if $\min(V) \le -\max(U)$, then each negative counter value is losing.
\item If $\max(U) > -\min(V)$, and if there exists a bound above which each counter value is winning,
then each counter value is winning.
The same holds if $\min(U) < -\max(V)$, and if there exists a bound below which each counter value is winning.
\end{itemize}
\end{proposition}

\begin{proof}
\begin{itemize}
\item We consider a robot game $(U,V)$ in which we have $\max(V) \ge - \min(U)$.
For any positive counter value $x$ and all moves $v_1, \dots, v_k \in V$, $u_1, \dots, u_k \in U$,
the opponent wins by playing the strategy $x v_1 u_1 \dots v_k u_k \mapsto \max(V)$:
every round ends in a counter value that is greater than or equal to the previous one, no matter what the reacher does.
The case where $\min(V) \le - \max(U)$ is analogous for negative counter values.
\item (First case only, the second one is analogous)
We consider a robot game $(U,V)$ for which we have $\max(U) > - \min(V)$,
and for any counter value $y$ above a certain $x \in \Z$,
the reacher has a winning strategy $\sigma_y$.
Here is the winning strategy for the reacher from any initial counter value:
In a play prefix where no counter value above $x$ has been visited, he plays $\max(U)$;
in a play prefix ending at $z$ where the first counter value above $x$ visited is $y$, he plays $\sigma_y(z)$.
Because $\max(U) > - \min(V)$, after every round the counter value visited grows
until it goes over $x$ where the reacher will win afterwards. Like in the proof of Proposition~$1$,
the reacher knows whether the first case or the second one is the right one and what the value of $y$ is.
\end{itemize}
\end{proof}

\section{The complexity of one-player robot games on the integer line}

When there is only one player in a robot game, i.e., $V = \{0\}$, the order of the moves does not matter.
The reacher has a winning strategy if, and only if, for each move $u \in U$,
there exists a number of times the reacher must use $u$, that is to say,
the negative of the initial counter value is a positive linear combination of moves in the set $U$.
We thus make a link between one-player robot games and linear programming.

\label{thm1}\begin{theorem}
Given a robot game in dimension one with $V = \{0\}$ and an initial counter value $x_0$,
deciding whether the reacher has a winning strategy from $x_0$ is NP-complete.
\end{theorem}

\begin{proof}
\ \\
\begin{itemize}
\item\kw{NP-membership}: The decision of the winner in a one-player robot game reduces to
the following integer linear programming problem, which is in NP according to \cite[p. $320$, Th.~$13.4$]{PS82}.
 \begin{alignat*}{4}
\text{Minimize~} x \quad &   \\
\text{subject to} \hspace*{1.6em}  &  ~x + \sum_{u \in U} a_u u & {}={} &  -x_0\\
&  \hspace*{4em} x & {}\ge{} &  0\\
& \hspace*{4em} a_u & {}\ge{} & 0,~ \hspace*{2em} u \in U
\end{alignat*}
The minimal $x$ is $0$ if, and only if, the reacher has a winning strategy.

\vspace*{5mm}

\item\kw{NP-hardness}: We present a polynomial-time reduction from the NP-complete \textsc{Subset-Sum} problem~\cite[p. 1097]{CLRS09}
to the decision of the winner in a one-player robot game.
For a given set $\{x_0,\dots,x_{n-1}\}$ of positive integers and a given positive integer $s$, the \textsc{Subset-Sum} problem
asks whether there exists a subset $I$ of $\llbracket 0,n-1\rrbracket$ such that $\sum_{i \in I} x_i = s$.
Let $(X,s)$ be an instance of \textsc{Subset-Sum}. Let $n = |X|$, $b = \max(X)$, and
$k = \lfloor\log_2(\max(nb,s))\rfloor+1$.
We build a robot game where we write counter values as their binary encoding, so we deal with bits in the following.

The basic idea of the reduction is that we start from $s$
and give to the reacher the possibility to subtract some $x_i$ and try to reach $0$, but it is not enough.
To prevent him from subtracting twice the same $x_i$, we add a control on the highest bits of the counter value.
Therefore, the initial counter value in the robot game is $s + \sum_{i=k}^{k+n-1} 2^i + n \cdot 2^{k+n}$,
and every reacher move subtracts at least $2^{k+n}$ from the counter value in order that at most $n$ moves can be performed.
More precisely, $U = \cup_{i=0}^{n-1} U_i$, where $U_i = \{-x_i - 2^{k+i} - 2^{k+n},-2^{k+i}-2^{k+n}\}$.

We now explain the link between a potential solution to the instance of \textsc{Subset-Sum} and a strategy in the robot game.
Consider a subset $I$ of $\llbracket 0,n-1\rrbracket$.
With at most $n$ moves from the set $U$, the only possibility to reset the bits number $k$ through $k+n-1$ in the robot game
is to decrement all of them once, i.e., use exactly one move from each set $U_i$.
We choose in the set $U_i$ the move $-x_i - 2^{k+i} - 2^{k+n}$ if $i$ is in $I$ and $-2^{k+i}-2^{k+n}$ else.
After the $n$ moves, the counter value is $0$ if, and only if, $\sum_{i \in I} x_i = s$.
Consequently, it is equivalent to find a winning strategy in the robot game and to find a subset of $X$ that sums up to $s$.
\end{itemize}
\end{proof}

\section{The complexity of two-player robot games on the integer line}

In this section, we present the tools to build an exponential-time algorithm that decides the winner in a robot game.
First, we explain the notion of an attractor, then we define the Frobenius problem and we give an over-approximation
of the solution to this problem, in order to find bounds above and below which we are sure that the same player always wins.
The algorithm in the third part computes the attractor and uses the bounds we get to avoid infinite recursion.
At the end of the section, we prove the matching lower complexity bound for the decision problem.

\subsection{The attractor construction}

We first define the one-step attractor of a set.
Consider a graph $(Q,Q_\exists,Q_\forall,E)$ for a general reachability game,
where~$Q$ is a possibly infinite set of vertices partitioned into
subsets~$Q_\exists$ for the reacher and~$Q_\forall$ for the opponent,
and~$E \subseteq Q \times Q$.
The \kw{one-step attractor} of a subset~$X$ of~$Q$, written~$\attrone(X)$, is the set of states
from which the reacher can force to go to~$X$ in one step, which means:
\begin{align*}
\attrone(X) = &\ \{q \in Q_\exists \text{ such that } \exists q' \in X, (q,q') \in E\}\\
& \cup \{q \in Q_\forall \text{ such that } \forall q' \in Q, (q,q') \in E \text{ implies } q' \in X\}.
\end{align*}

The \kw{attractor} of $X$, written $\attr(X)$, is the set of states
from which the reacher has a strategy to eventually go to $X$ no matter what the opponent plays,
in other words he has a winning strategy in the reachability game with objective $X$ on the aforementioned arena.
The set $\attr(X)$ is the least fixpoint of $\attrone$ containing $X$.
We obtain it recursively: compute $Y = X\ \cup\ \attrone(X)$, if $Y = X$ then return $Y$ else set $X := Y$ and repeat.

Let us adapt a robot game to these notations.
The reacher owns $Q_\exists := \{\circpl\} \times \Z$ and the opponent owns $Q_\forall := \{\boxpl\} \times \Z$.
The set of edges is the union of the set $\{((\boxpl,x),(\circpl,y))\ |\ x, y \in \Z, y-x \in V\}$,
which represents the opponent moves, and of the set $\{((\circpl,x),(\boxpl,y))\ |\ x, y \in \Z, y-x \in U\}$,
which represents the reacher moves. The objective for the reacher is the vertex $(\boxpl,0)$.
In our definition of robot games, winning positions are counter values.
They are here represented as a pair ($\boxpl$ or $\circpl$ , the counter value),
but we only care for winning positions with $\boxpl$ as left component when we solve the game,
because a play starts with the opponent.

We use here two-step attractors $\attrtwo(X) = \attrone(\attrone(X))$, rather than one-step attractors,
because of the round-based structure of a play in the robot game.
The winning set in a robot game is $\attr(\{(\boxpl,0)\})$.
We call it trivial if its intersection with the opponent vertices is restricted to $\{(\boxpl,0)\}$,
which is the case if, and only if, the computation of $\attr(\{(\boxpl,0)\})$ stops at the second step
because a fixpoint has already been reached.
In other words, the winning set in a robot game is trivial if, and only if,
the set $\attrtwo(\{(\boxpl,0)\})$ is either empty or $\{(\boxpl,0)\}$.

\label{prop3}\begin{proposition}
The winning set in a robot game $(U,V)$ is non-trivial if, and only if,
there exists a counter value $x \not= 0$ such that, for all opponent moves $v \in V$,
there is a reacher move $u \in U$ such that $u + v = -x$.
\end{proposition}

\begin{proof}
Given a counter value $x \not= 0$, a configuration $(\boxpl,x)$ is in $\attrtwo(\{(\boxpl,0)\})$
if, and only if, for any opponent move $v \in V$, we have $(\circpl,x+v) \in \attrone(\{(\boxpl,0)\}$,
and this is equivalent to the existence of a $u \in U$, which depends on $v$ and $x$, such that $x+v+u = 0$.
\end{proof}

Let us look at the game presented in the Figure~$1$. Here, consider a play that starts at $-3$:
if the opponent chooses to play $3$, then the reacher wins by playing $0$, if the opponent plays $-1$,
then the reacher wins by playing $4$.
With the terminology of this section, it means that $(\boxpl,-3) \in \attrtwo(\{(\boxpl,0)\})$.

We define an integer version of $\attrtwo$, for a subset $X$ of $\Z$, by
$$\cpre(X) = \{x \in \Z\ |\ (\forall v \in V)(\exists u \in U)\ x + u + v \in X\}.$$

Note that a round begins in $\cpre(X)$ if, and only if, the reacher can force this round to end in $X$.
Let $X$ be a subset of $\Z$, and let $X_{\text{opp}} = \{\boxpl\} \times X$.
Because the predecessor of a vertex with $\boxpl$ in the left component
can only be a vertex with $\circpl$ in the left component and vice-versa,
we have $\attrtwo(X_{\text{opp}}) = \{\boxpl\} \times \cpre(X)$.

The next result is very important for our algorithm.
We will usually be in a situation where the algorithm computes a bound $b$ such that
we can decide immediately for which player a counter value $x$ that is greater in absolute value than $b$ is winning.
The proposition presents the bounded arena we build from the robot game, where termination is guaranteed for the computation of the attractor.

\label{prop4}\begin{proposition}
Consider a robot game $G$ for which there exist two integers $d \in \N \setminus \{0\}$ and $b \in \N$ such that
no negative counter value is winning and every counter value greater than $b$
is winning if, and only if, it is a multiple of $d$.
We can build a reachability game on a finite arena on which the reacher has a winning strategy
if, and only if, he has a winning strategy in $G$. 
\end{proposition}

\begin{proof}
Let Restr$^b_d(U,V) = (Q,Q_\exists,Q_\forall,E)$, where:
\begin{itemize}
\item $Q_\forall = \{\bot_{< 0}, \top_{> b}, \bot_{> b}\} \cup (\{\boxpl\} \times \llbracket 0,b\rrbracket)$;
\item $Q_\exists = \{\circpl\} \times \llbracket\min(V),b+\max(V)\rrbracket$;
\item $Q = Q_\exists \cup Q_\forall$;
\item $E = \{((\boxpl,x),(\circpl,y)) \in Q_\forall \times Q_\exists \ |\ y - x \in V\}\\
\cup \{((\circpl,x),(\boxpl,y)) \in Q_\exists \times Q_\forall \ |\ y - x \in U\}\\
\cup \{((\circpl,x),\bot_{< 0}) \in Q_\exists \times Q_\forall\ |\ \exists u \in U, x + u < 0\}\\
\cup \{((\circpl,x),\bot_{> b}) \in Q_\exists \times Q_\forall\ |\ \exists u \in U, x + u > b \wedge x + u \not\in d\Z\}\\
\cup \{((\circpl,x),\top_{> b}) \in Q_\exists \times Q_\forall\ |\ \exists u \in U, x + u > b \wedge x + u \in d\Z \}\\
\cup \{(\bot_{< 0},\bot_{< 0}),(\top_{> b},(\boxpl,0)),(\bot_{> b},\bot_{> b})\}$.
\end{itemize}

The reachability game played on Restr$^b_d(U,V)$ where the reacher wants to go from $(\boxpl,x)$ to $(\boxpl,0)$,
for a given $1 \le x \le b$, is actually the robot game $(U,V)$ with the initial value $x$,
in which we stop the play as soon as we know the winner.
Indeed, we supposed that all negative counter values are losing,
that is why, on Restr$^b_d(U,V)$, instead of vertices $(Q_\forall,x)$ for $x \in -\N$
we have a losing sink $\bot_{< 0}$. Similarly, because in $(U,V)$ all counter values above b
are winning if, and only if, they are multiples of $b$, on Restr$^b_d(U,V)$,
instead of vertices $(Q_\forall,x)$ for $x > b$, we have two sinks, one winning and one losing,
and the redirection of edges depend on the counter value.
\end{proof}

We allow the notation Restr$^b_d(U,V)$ for negative integers $b$,
given that no positive counter value is winning
and every counter value less than $b$ is winning if, and only if,
it is a multiple of $d$. In fact Restr$^b_d(U,V)$ is the arena $Restr^{-b}_d(-U,-V)$.
Given $d$ and $b$, we can decide the winner in the reachability game on Restr$^b_d(U,V)$
using the attractor construction, because this time the arena is finite.
We write $\restrattr(G) = \{x \in \Z\ |\ (\boxpl,x) \in \attr(\{(\boxpl,0)\})\}$
where $\attr(\{(\boxpl,0)\})\}$ is the winning set in the game described above on the arena $G =$ Restr$^b_d(U,V)$.
The function $\restrattr$ is used in the main algorithm.

\subsection{The Frobenius problem}

Let $W$ be a non-empty subset of $\Z$. The arithmetical notions we present in this section
are part of the algorithm: $W$ stands for a subset of the winning counter values.
We denote by gcd$(W)$ the greatest common divisor of $W$,
and we compute it as follows: gcd$(\{d\}) = d$, and for $W \not= \emptyset$, gcd$(\{w\} \cup W)$ is the usual greatest common divisor
of $w$ and gcd$(W)$. The integers in $W$ are mutually prime if gcd$(W) = 1$.

The \kw{Frobenius problem} asks for the greatest integer that is not $W$-reachable,
where $W$ is a set of mutually prime positive integers.

Note that the set of non-$W$-reachable positive integers would be infinite without the assumption of mutual primality.
It is empty whenever the set $W$ contains the value $1$, in which case the solution to the Frobenius problem is $-1$, by convention.
Theorem~$6$, which follows from \cite{Wil78}, gives a bound to the solution to the Frobenius problem for a given set. 

\label{thm2}\begin{theorem}
Let $W$ be a set of mutually prime positive integers.
The solution to the Frobenius problem for $W$ is less than or equal to $\max(W)^2$.
\end{theorem}

Here, we are interested in a variant of the Frobenius problem on arbitrary subsets $W$ of $\N$ or $-\N$,
where we look for a bound beyond which every integer is $W$-reachable if, and only if, it is a multiple of gcd$(W)$.
When $W$ is a set of mutually prime positive integers, this is exactly the Frobenius problem.
Otherwise, let $W \subseteq \N$ and $d =$ gcd$(W)$. Consider the set $W' = \{\frac{w}{d}\ |\ w \in W\}$,
which contains mutually prime positive integers.
Let $F$ be the solution to the Frobenius problem for $W'$.
Then the set of $W$-reachable integers greater than $dF$ is equal to the set of multiples of $d$ greater than $dF$.
Consequently, the solution to our problem for $W$ is $dF$.
For $W \subseteq -\N$, we procede analogously.

Actually, the computation of $F$ is hard and the bound in Theorem~$6$ has at most twice the size of~$W$. 
That is why we use the following function in the algorithm.
Let $W \subseteq \N$ such that gcd$(W) = 1$. We write $\tilde{F}(W) = \max(W)^2$.
Let $W \subseteq -\N$ such that gcd$(W) = 1$. We write $\tilde{F}(W) = -\max(-W)^2 = -\min(W)^2$, under which all integers are $W$-reachable.
We extend $\tilde{F}(W)$ when gcd$(W) = d \not= 1$:
Let $W' = \{\frac{w}{d}\ |\ w \in W\}$, we set $\tilde{F}(W) := d\tilde{F}(W') =\max(|W|)^2/d$.
In the particular case where $W$ is a singleton, we set $\tilde{F}(W) = 0$.

Finally, when $W$ is neither included in $\N$ nor in $-\N$, we decide $W$-reachability according to the following lemma.

\label{lemma1}\begin{lemma}
Let $W$ be a finite subset of $\Z$ that has two elements of opposite signs.
An integer is $W$-reachable if, and only if, it is a multiple in $\Z$ of gcd$(W)$.
\end{lemma}

\begin{proof}
Consider two elements $w > 0$ and $w' < 0$ of $W$.
The integers $-w$ and $-w'$ are $W$-reachable
because $-w = (-w'-1)w + w w'$ and $-w' = (w-1)w' + (-w')w$,
which are combinations with only nonnegative coefficients.

By B\'{e}zout's identity, there exist integer coefficients $a_w$ such that $\sum_{w \in W} a_w w = $ gcd$(W)$.
We replace $a_w w$ by $(-a_w) \cdot (-w)$ for all negative coefficients $a_w$.
The resulting linear combination has only positive coefficients, therefore gcd$(W)$ is $W$-reachable, as well as $-$gcd$(W)$.
We conclude that $W$-reachability is equivalent to membership in gcd$(W)\Z$.
\end{proof}

\subsection{A theorem by Sylvester}

This section aims at giving an alternative way to bound the solution to the Frobenius problem.
Using Theorem~$6$ is simpler, but we can have a sharper bound. 

Theorem~$8$ relies on the extended Euclidean algorithm, presented in \cite[p. 937]{CLRS09}. 
With an iteration of the algorithm to more than two integers
according to the way we present the gcd of a set, we can prove Corollary~$9$. 

\label{thm5}\begin{theorem}
Let $a,b$ be two integers. B\'{e}zout coefficients for $a$ and $b$, that is to say integers $u,v$ such that $ua + vb =$ gcd$(a,b)$,
can be computed with a time complexity polynomial in the size of the binary encoding of $a$ and $b$.
\end{theorem}

\label{cor2}\begin{corollary}
B\'{e}zout coefficients for a finite subset of $\Z$ are computable
with a polynomial time complexity in the size of the binary encoding of the integers in the subset.
\end{corollary}

The article \cite{PRS05} mentions a theorem concerning the Frobenius problem, due to Sylvester in \cite{Syl1882}.

\label{thm6}\begin{theorem}
Let $W = \{p,q\}$ be an instance of the Frobenius problem. Then the greatest non-$W$-reachable integer is $pq - p - q$.
Moreover, an integer $x$ is reachable if, and only if, $pq - p - q - x$ is not,
which means exactly half of the integers between $0$ and $pq - p - q$ are reachable.
\end{theorem}

Such a simple statement does not extend well when $W$ has more than two elements.
If there are two mutually prime integers in $W$, then we can take them, else
we have to find two mutually prime $W$-reachable integers to get an upper bound of the maximal non-$W$-reachable integer.
For example, if $W = \{6,10,15\}$, then there is no pair of mutually prime integers in $W$ even though gcd$(W) = 1$.
Nevertheless, $25$ is $W$-reachable and we have gcd$(6,25) = 1$, hence every integer above $6 \cdot 25 - 6 - 25$ is $W$-reachable.
In any case, Proposition~$11$ gives a way to apply the above theorem 
and there is only a finite number of integers for which we cannot find out immediately whether they are $W$-reachable or not.

\label{prop9}\begin{proposition}
Let $W$ be a subset of $\Z$ such that gcd$(W) = 1$.
There is a polynomial time algorithm that gives a pair of mutually prime $W$-reachable integers.
\end{proposition}

\begin{proof}
If two integers in $W$ are mutually prime, then there is nothing to do.
Else consider the linear combination, obtained with the extended Euclidean algorithm,
$\sum_{i = 1}^n a_i w_i = 1$ for $w_1, \dots, w_n \in W$.
We suppose that the terms are ordered such that for a certain $1 \le k \le n+1$ all $a_i, i < k$ are positive
and all $a_i, i \ge k$ are negative.
First note that $n \ge 3$, else two integers in $W$ would be mutually prime.
Let us distinguish three cases:
\begin{itemize}
\item If $k = 1$, in other words all $a_i$ are negative,
then consider the $W$-reachable integers $p := w_1$ and $q := \sum_{i = 2}^n (-a_i) w_i$.
We apply B\'{e}zout's theorem: $p$ and $q$ are mutually prime because $a_1 p - q = 1$.
\item Similarly, if $k = n+1$, in other words all $a_i$ are positive,
then the $W$-reachable integers $p := w_1$ and $q := \sum_{i = 2}^n a_i w_i$ are mutually prime.
\item Else, the $W$-reachable integers $p := \sum_{i = 1}^{k-1} a_i w_i$ and $q := \sum_{i = k}^n -a_i w_i$
are mutually prime and defined by a non-empty sum.
\end{itemize}
\end{proof}

\subsection{The algorithm}

We have now all necessary tools to solve robot games.
The main idea is to iterate the computation of $\cpre$
until we establish that we can describe the winning set with the finite set obtained so far. 

We prove in Proposition~$12$ and its corollary that, for $X$ the set that we compute in the first step of our algorithm
and Win the winning set in the robot game,
$X \subseteq$ Win and that, for a well-chosen set $Y'$ of $X$-reachable counter values, 
if gcd($\cpre(Y')$) $= \text{gcd}(X)$, then also gcd(Win) $= \text{gcd}(X)$.
Basically, the first step relies on this property: We compute successive $\cpre$, and once the step ends we get gcd(Win).
Actually, to keep control over the complexity, we do not do $X := X\ \cup\ \cpre(X)$,
but only add to $X$ a single element $y$ of the computed $\cpre$
such that gcd$(X) \not=$ gcd$(X \cup \{y\})$.

Once we find gcd(Win), there are two cases. In the first case, Lemma~$7$ can be applied and we are done, 
because Win has two elements with opposite signs. Hence, the winning set is gcd(Win)$\Z$.
In the second case, the winning set is included in one of the sets $\N$ or $-\N$;
we suppose without loss of generality that the winning set is included in $\N$.
Theorem~$6$ yields a bound above which 
all multiples of gcd(Win) and only them
are winning because they are $X$-reachable.
Therefore, the only set of counter values about which we still do not know whether they are winning or not is empty or bounded and,
by Proposition~$5$, we can compute an attractor on the restricted arena. 

\label{prop5}\begin{proposition}
Let Win be the winning set in a robot game $(U,V)$,
and $d \in \N$ be a multiple of gcd(Win) that is not gcd(Win). Let $Y = d\Z\ \cap$ Ampl$^d(U,V)$.
Then we have $\cpre(Y) \setminus d\Z \not= \emptyset.$
\end{proposition}

\begin{proof}
First, we establish that if $\cpre(d\Z)$ is not included in $d\Z$,
then neither is $\cpre(d\Z\ \cap$ Ampl$^d(U,V))$.
Let $x \in \cpre(d\Z) \setminus d\Z$.
All counter values $x + v + u$ for $v \in V$ and $u \in U$ are included in the interval Ampl$(U,V) + x$,
a fortiori when $u$ is chosen according to $x$ and $v$ such that $x + v + u \in d\Z$.
In this case, $x$ mod $d$ belongs to $\cpre(d\Z\ \cap$ Ampl$^d(U,V)) \setminus d\Z$, and it is outside $d\Z$ too.

Second, we prove the proposition by contrapositive:
Suppose that $\cpre(d\Z\ \cap$ Ampl$^d(U,V))$ is included in $d\Z$.
We just proved that it implies the inclusion of $\cpre(d\Z)$ in $d\Z$.
As a consequence, from any counter value outside $d\Z$, there exists an opponent move such that for all reacher moves,
the next round begins outside $d\Z$ too, in particular it is impossible for the reacher to have a winning strategy.
Hence, $d$ divides gcd(Win).
\end{proof}

We need to adapt this result because we do not know whether $Y \subseteq Win$
and we look for a statement that allows us to find winning counter values.
That is why we define the regularity interval $I_{(U,V)}(X)$ of a finite subset $X$ of $\Z$ neither empty nor equal to $\{0\}$ in a robot game $(U,V)$.
The elements of this interval are $X$-reachable if, and only if, they are multiples of gcd($X$).
\begin{itemize}
\item If $X \subset \N$, then
$I_{(U,V)}(X) := (\tilde{F}(X) - \min(V) - \min(U) + d) + \text{ Ampl}^d(U,V)$,
the lower bound of this interval is $\tilde{F}(X)$.
\item If $X \subset -\N$, then
$I_{(U,V)}(X) := (\tilde{F}(X) - \max(V) - \max(U) - d) + \text{ Ampl}^d(U,V)$,
the upper bound of this interval is $\tilde{F}(X)$.
\item Else, $I_{(U,V)}(X) :=$ Ampl$^d(U,V)$.
\end{itemize}

\label{cor1}\begin{corollary}
Let Win be the winning set in a robot game $(U,V)$,
and $X \subset$ Win such that gcd$(X) = d >$ gcd(Win). Let $Y' = I_{(U,V)}(X) \cap d\Z$.
Then we have $\cpre(Y') \setminus d\Z \not= \emptyset.$
As a consequence, if Win $\not\subseteq d\Z$,
then we can compute a certain element of the difference in space polynomial in~$|U|$ and~$|V|$.
\end{corollary}

We apply this idea to the game in Figure~$1$ and find that $-2$ is a winning counter value outside $-3\N$,
because if the opponent plays $3$, then the reacher can play $-1$ and win,
and if the opponent plays $-1$, then the reacher can play $0$,
and in the next round he can play the difference between $3$ and the opponent move to win.
With the notations of the last proposition and of its corollary, we have $\tilde{F}(\{-3\}) = 0$,
the interval $I_{(\{-1,0,4\},\{-1,3\})}(\{-3\})$ is $(0 - 3 - 4 - 3) + \llbracket -5,10\rrbracket = \llbracket -15,0\rrbracket$, and $\cpre(\{-15,-12,\dots,0\})$
is not included in $3\Z$. We pick $-2$ in it. Since gcd($\{-2,-3\}$) $= 1$, we know that gcd(Win) is $1$.

\begin{algorithm}
\caption{Algorithm for solving robot games on the integer line.}
\KwData{A robot game $(U,V)$.}
\KwResult{A description of the winning set.}
{\tt /* Require: Functions computing the sets we use, as defined in the Section~$4$. */}\;
\Begin{
$d\gets 0$\;
$X\gets \cpre(\{0\}) \cup \{0\}$
{\tt /* to avoid handling gcd$(\{0\})$ in the first step */}\;
  \lIf{$X = \{0\}$}{\Return{$X$}}\;\;
{\tt /* Step \textbf{1.} */}\;
\lWhile{$d = 0$}{\;
\Indp
  $d'\gets$ gcd$(X)$\;
  $I\gets I_{(U,V)}(X)$\;
{\tt /* $I$ is a set of $X$-reachable counter values with a large absolute value */}\;
  $Y\gets \cpre(I\ \cap d'\Z)$\;
{\tt /* From $Y$, the reacher can force the next round to end at a counter value known to be winning */}\;
  \lIf{$Y \setminus d'\Z \not= \emptyset$}{$X\gets X \cup \{\min(Y \setminus d'\Z)\}$}
{\tt /* minimum in absolute value */}\;
  \lElse{$d \gets d'$}
{\tt /* We know that $d$ is gcd(Win): we exit the loop */}\;\;
\Indm
}\;
{\tt /* Step \textbf{2.} */}\;
  \lIf{$X \not\subseteq \N \wedge X \not\subseteq -\N$}{\Return{$d\Z$}}
{\tt /* Lemma~$7$ */}\; 
  \lElse{\;
\Indp
  $I\gets$ Ampl$(U,V)$\;
  $b\gets \tilde{F}(X)$\;
    \lIf{$X \subseteq \N$}{\;
  \Indp
      \lIf{$-\N\ \cap \cpre(I \cap d\N) \not= \emptyset$}{\Return{$d\Z$}}
{\tt /* Lemma~$7$, second try */}\; 
      \lElse{$Unbd\gets \{x \in d\Z\ \mid\ x > b\}$}
{\tt /* Half-line of winning counter values */}\;
  \Indm
      }
    \lElse{\;
  \Indp
      \lIf{$\N\ \cap \cpre(I \cap -d\N) \not= \emptyset$}{\Return{$d\Z$}}\;
      \lElse{$Unbd\gets \{x \in d\Z\ \mid\ x < b\}$}\;
  \Indm
      }
  $G\gets$ Restr$^b_d(U,V)$\;
{\tt /* Between $0$ and $b$, we compute the attractor on the restricted arena according to Proposition~$5$ */}\;
  \Return{$Unbd\ \cup \restrattr(G)$}\;
\Indm
  }
}
\end{algorithm}

\label{thm3}\begin{theorem}
Algorithm~$1$ computes the winning set in a robot game in exponential time.
\end{theorem}

\begin{proof}[Proof of termination]
The only loop in the algorithm is in the first step.
Each iteration either lowers the variable $d'$, more precisely replaces it by one of its divisors,
or assigns the variable $d$ to the value of $d'$, which makes the loop stop because this value is positive.
The lowering of $d'$ occurs less times than the gcd of $\cpre(\{0\}) \cup \{0\}$,
and if this set is $\{0\}$, then the algorithm stops before the first step begins.
\end{proof}

\begin{proof}[Proof of correctness]
Let Win be the actual winning set, and let $d =$ gcd(Win).
\begin{itemize}
\item In the first step, the variable $X$ is a subset of Win. We prove it by recurrence:
\begin{itemize}
\item The step begins with $X = \cpre(\{0\}) \cup \{0\}$, which contains only winning counter values.
\item Let $X \subseteq$ Win, let $d' =$ gcd$(X)$.
In the loop, when a counter value $y$ is included in $X$,
it belongs to $\cpre(I \cap d'\Z)$, where $I$ is the regularity interval of $X$.
Thus the reacher has a move to go from $y$ to a subset of $X$-reachable counter values,
which justifies that $y$ is a winning counter value too.
\end{itemize}
\item On the other hand, if no element of Win$\setminus d\Z$
is found and included in $X$, then by Corollary~$13$, there exists none. 
It remains to look for elements of Win$\setminus \langle X\rangle_\N$, necessarily in $d\Z$.
\item We distinguish three cases to prove the second step.
\begin{itemize}
\item If two counter values in $X$ have opposite signs, then Win $= \langle X\rangle_\N = d\Z$ by Lemma~$7$. 
\item Else if $X \subseteq \N$ and $-\N\ \cap \cpre($Ampl$(U,V) \cap d\N) \not= \emptyset$,
then we also have Win $= d\Z$. Actually we here prove this equivalent to the fact that
two counter values in Win have opposite signs.

$(\Leftarrow)$ Let $x_0 \in -\N\ \cap $ Win.
Consider a play $\pi$ that starts at $x_0$ and in which the reacher uses a winning strategy.
The play $\pi$ ends in $0$ and every round finishes in winning counter values, i.e., multiples of $d$.
Let $x \in -\N$ be the counter value in which a round in $\pi$ ended and no more round ended in $-\N$ afterwards.
Whatever the opponent did, the reacher forced the round that began in $x$ to end in a nonnegative winning counter value.
To sum up, $x$ is a negative counter value in $\cpre($Ampl$(U,V) \cap d\N)$.

Note that Ampl$(U,V)$ necessarily contains negative counter values, else there would not be any positive winning counter value.

$(\Rightarrow)$ Let $x \in -\N\ \cap \cpre($Ampl$(U,V) \cap d\N)$.
In other words, for every opponent move, the reacher has a move such that a round that begins in $x$
ends in a positive multiple of $d$ in one round, and this multiple is less than $\max(U)+\max(V)$.
Consider the reacher move as the image of the opponent move by a function $\varphi : V \to U$.
If the reacher plays the image by $\varphi$ of the last opponent move $dk$ times,
for $k \in \N$ big enough, then a great multiple of $d$, i.e., a counter value in $\langle X\rangle_\N$, is reached.
This justifies that $x$ is winning.

If $X \subseteq -\N$ and $\N\ \cap \cpre($Ampl$(U,V) \cap -d\N) \not= \emptyset$,
we have the same result.

\item Else, we know that all counter values in Win have the same sign.
Suppose without loss of generality that $X\subseteq \N$.
From a negative counter value, only negative or positive but surely losing counter values can be visited,
therefore Win is included in $d\N$.
We use Theorem~$6$: every counter value above $\tilde{F}(X)$ is winning. Between $0$ and $\tilde{F}(X)$, 
we decide the winner using the result of Proposition~$5$ about the attractor on the restricted arena. 

\end{itemize}
\end{itemize}
\end{proof}

\begin{proof}[Proof of complexity]
We consider the input size as $\sum_{w \in U \cup V} \log(|w|)$.
The algorithm first computes the set $\cpre(\{0\})$, which contains at most $|U|$ counter values, 
all of them have a lower size than the input size.

Let us consider the loop in the first step of the algorithm.
For any subset $X$ of $\Z$, let $d'$ be the gcd of $X$ and let $I$ be the regularity interval of $X$.
The size of $I$ is $2$ gcd$(X) + \min(U) + \max(U) + \min(V) + \max(V)$, one of its bounds is $\tilde{F}(X)$,
and the size of a representation of this integer is at most twice the size of $X$.
The size of the counter value $y$ obtained in the loop using $\cpre$ on $I \cap d'\Z$
is bounded by a polynom in the size of the integers in $X$.
There is a logarithmic number of iterations in the loop, because each assignment of $d'$
sets it to one of its strict divisors.

We now look at the second step of the algorithm.
It first checks whether two integers in $X$ have opposite signs,
and in case of fail makes another test on the $\cpre$ of an interval included in the amplitude of the game.
This can be done in polynomial time.
If the second test fails, then the bound $b := \tilde{F}(X)$ is computed,
the arena $G :=$ Restr$^b_d(U,V)$ is built and the reachability game on $G$ is solved with the computation of an attractor,
for a time complexity polynomial in the size of $G$.
This size is linear in the value of $b + \max(V) - \min(V)$.
With a binary encoding, the algorithm uses then exponential time.
\end{proof}

Let us illustrate the second step of the algorithm with the example in Figure~$1$ again.
We exit the first step with a subset $X = \{-2,-3\}$ of the winning set such that gcd$(X) = 1$.
Because $1 = -\min(U) < \max(V) = 3$, the opponent wins from any positive counter value (Proposition~$2$), 
it is indeed impossible that $\cpre(\{-2, -1, 0\})$ contains any positive counter value.
Every nonpositive $\{-2,-3\}$-reachable counter value, i.e., every nonpositive counter value but $-1$, is winning.
We only have to decide whether the reacher wins from $-1$, and it is not the case because the opponent can play $3$ every time,
which guarantees that only positive counter values are visited after the first move.
The algorithm decides it when it calls $\restrattr$ on the arena Restr$^{-1}_1(U,V)$.

\subsection{The lower complexity bound}

We are now showing EXPTIME-hardness of robot games. In order to do this, we give the definition of countdown games \cite{JLS07},
which are games with one positive and strictly decreasing counter.
We then introduce a variant of countdown games and show two successive reductions from countdown games to our variant and then from this variant to robot games.

A \kw{countdown game} between two players $1$ and $2$ is represented by a pair $((S,T),c_0)$
where $S$ is a finite set of locations, $T \subseteq S \times (\N \setminus \{0\}) \times S$ is a set of weighted transitions and $c_0 \in \N \setminus \{0\}$.
We consider that $S$ has a particular location~$s_0$.
Configurations in the game are pairs $(s,c) \in S \times \N$, where $c$ is a counter value.
A play is a sequence of moves, done in the following way: from a configuration $(s,c)$, initially $(s_0,c_0)$,
player~$1$ chooses a value $d \le c$ called duration such that there exists a transition in $T$
with $s$ as first component and $d$ as second component, then player~$2$ chooses $(s,d,s')$ among these transitions.
This move updates the configuration to $(s',c-d)$.

The winner of a play in a countdown game is determined once the play is blocked.
Because only nonnegative integers appear in the configurations and positive integers in the transitions,
the game is finite and ends when player~$1$ cannot find any duration to make a move.
At this point, player~$1$ wins if, and only if, the counter value is~$0$.
Deciding the winner in countdown games is EXPTIME-complete. \cite{JLS07}

We now define \kw{restricted countdown games}: On the one hand,
the winning condition for player~$1$ is now that the play ends in $(\bot,0)$ for a particular sink $\bot \in S$,
i.e., there are no transitions with $\bot$ as first component;
on the other hand, if there are two transitions $(s_1,d,s'_1)$, $(s_2,d,s'_2)$ in $T$, then $s_1 = s_2$;
in other words, a duration is specific to a location.

\label{prop6}\begin{proposition}
Countdown games reduce in polynomial time to restricted countdown games.
\end{proposition}

\begin{proof}
There are two steps in the construction.
Consider an arbitrary countdown game $G = ((S,T),c)$, where $S = \{s_0, \dots, s_{n-1}\}$.
Let $d'$ be the least positive integer that does not appear in any transition in $T$.
First, we build the countdown game $G' = ((S \cup \{\bot\}),T \cup \{(s,d',\bot)\ |\ s \in S\}, c+d')$,
The winning condition for player~$1$ in $G'$ is to reach $(\bot,0)$.

Notice that player~$1$ wins a play $\pi$ in $G'$ if, and only if, in the last move of $\pi$, the configuration is $(s,d')$
for any location $s$, and player~$1$ chooses $d'$, in order that player~$2$ can only pick the transition $(s,d',\bot)$.
The partial play from $(s_0,c+d')$ to $(s,d')$,
corresponds in $G$ to a play that starts at $(s_0,c)$ and ends in $(s,0)$, where player~$1$ wins.

Second, we build from $G'$ a restricted countdown game that we prove equivalent to $G$.
Let $G'' = ((S \cup S' \cup \{\bot\},T''),2N(c+d'))$, where $S' = \{s_1',\dots,s_{N-1}'\}$ and
$T'' = \{(s_0,2nd,s)\ |\ (s_0,d,s) \in T'\} \cup \{(s_i,i,s_i'), (s_i',2nd-i,s_j)\ |\ (s_i,d,s_j) \in T'\}$.
Matching transitions are $(s_0,2nd,s) \in T''$ and $(s_0,d,s) \in T'$,
as well as $(s_i',2nd-i,s_j) \in T''$ and $(s_i,d,s_j) \in T'$.
Matching plays are $\pi'' \in T''$ and $\pi' \in T'$ such that, when we exclude the moves from $s \in S$ to $s'$ in $\pi''$,
the transitions of every move in $\pi''$ and in $\pi'$ match.

The game $G''$ is a restricted countdown game
because the duration of a transition that starts in $s_i \in S$ is a multiple of $2N$ plus $i$
and the duration of a transition that starts in $s_i' \in S'$ is a multiple of $2N$ minus $i$.

Moreover, player~$1$ wins in $G''$ if, and only if, he wins in $G'$, thus in $G$.
Indeed, consider matching plays $\pi''$ of $G''$ and $\pi'$ of $G'$.
When the location is in $S \cup \{\bot\}$, the counter value in $\pi''$ is $2n$ times the counter value in $\pi'$,
because at the beginning of $\pi''$ the configuration is $(s_0,2n(c_0+d'))$
and at the beginning of $\pi'$ the configuration is $(s_0,c_0+d')$.

Hence, a play in $G''$ reaches $(\bot,0)$ if, and only if, the matching play in $G'$ reaches $(\bot,0)$.
\end{proof}

\label{thm4}\begin{theorem}
Given a robot game on $\Z$ and an initial counter value,
deciding whether the reacher has a winning strategy from this counter value is EXPTIME-hard.
\end{theorem}

We prove the theorem by a reduction from restricted countdown games to unidimensional robot games.
Let $((S,T),c)$ be a restricted countdown game, where we suppose without loss of generality that $S = \llbracket 0,n-1\rrbracket,
s_0 = 0$ and $\bot = n-1$. We name $D = \{d_0,\dots,d_{h-1}\}$ the set of values that appear in $T$ as durations.
Let $k = \lfloor\log_4(c)\rfloor+1$ and $k' = \lfloor\log_4(k)\rfloor+1$.

We write the counter value in the robot game in base $4$
and with $h+n+k+k'+1$ digits, which we split in four parts.
These parts encode \kw{durations}, \kw{locations}, \kw{value} in the countdown game and \kw{controls}
from the least significant digit to the most significant one.
We explain these notions after the presentation of the sets of moves $U$ and $V$.
The initial counter value is $4^h + c_0 \cdot 4^{h+n} + k \cdot 4^{h+n+k} + 4^{h+n+k+k'}$,
it corresponds to no duration given, the initial state, the value $c_0$ in the countdown game, and a default value for the control part.

Let us use an example to see how we represent a configuration in a restricted countdown game
by the counter value in the robot game.
Consider the countdown game pictured in Figure~$2$.
We represent in Figure~$3$ the first move of a play that starts at~$0$ with value $8$ (first line)
and where player~$1$ chooses duration $3$ (second line), after which player~$2$ moves to the location~$1$ (third line).
We set $D = \{1,2,3,4,5,6\}$.

\begin{table}
\newcolumntype{A}{m{0.45\textwidth}}
\begin{tabular}{AA}
\begin{center}
\begin{tikzpicture}[->,>=stealth',shorten >=1pt,auto,node distance=2.8cm,
                    semithick]
  \tikzstyle{state}=[circle,minimum size=1cm,fill=white,draw=black,text=black]

  \node[state] (A)                    {$0$};
  \node[state] (B) [right of=A]       {$1$};
  \node[state] (C) [below of=A]       {$2$};
  \node[state] (D) [right of=C]       {$\bot$};
  
  \path (A) edge [loop left] node {$6$} (A)
            edge [bend left] node {$3$} (B)
            edge             node {$3$} (C)
        (B) edge             node {$2$} (A)
            edge [bend left] node {$2$} (C)
            edge             node {$1$} (D)
        (C) edge [bend left] node {$4$} (A)
            edge             node {$4$} (B)
            edge    [swap]   node {$5$} (D);

\end{tikzpicture}
\end{center}
&
\begin{center}
\begin{tabular}{| c | c | c | c |}
\hline
$\overbrace{0\ 0\ 0\ 0\ 0\ 0}^{\text{duration}}$ & $\overbrace{1\ 0\ 0\ 0}^{\text{location}}$
& $\overbrace{0\ \ \ \ \ \ \ \ \ 2}^{\text{countdown}}$ & $\overbrace{2\ \ \ \ \ \ \ \ \ 1}^{\text{control}}$\\
\hline
\end{tabular}
\\
\vspace{\baselineskip}
\begin{tabular}{| c | c | c | c |}
\hline
$0\ 0\ 1\ 0\ 0\ 0$ & $0\ 0\ 0\ 0$ & $1\ \ \ \ \ \ \ \ \ 1$ & $2\ \ \ \ \ \ \ \ \ 1$\\
\hline
\end{tabular}
\\
\vspace{\baselineskip}
\begin{tabular}{| c | c | c | c |}
\hline
$0\ 0\ 0\ 0\ 0\ 0$ & $0\ 1\ 0\ 0$ & $1\ \ \ \ \ \ \ \ \ 1$ & $2\ \ \ \ \ \ \ \ \ 1$\\
\hline
\end{tabular}
\end{center}
\\
\captionof{figure}{Restricted countdown game.}
&
\captionof{figure}{Counter values in the robot game.}
\end{tabular}
\end{table}

For simplicity, we decide that the reacher begins in the robot game,
but the winning condition is still that the counter value becomes $0$ after the turn of the reacher.
It remains computationally equivalent.

We first give the moves of both players in the robot games as codes,
in order to explain the way we encode a play in the restricted countdown game.
Intuitively, because we split the counter in four parts, a risk appears that the encoding
no longer corresponds to a configuration in the reduced countdown game, for example because of a carry.
We prove that if a player tries to create such a bad behaviour,
then the other one can react with a winning strategy.

In a restricted countdown game, let us write $s_d$ for the location uniquely determined by a duration $d$.

The codes for the set of opponent moves are $\{$\duration\ $d$ \goto\ $s'\ \mid\ (s_d,d,s') \in T \}$,
which correspond to the choice by player~$2$ of the next location according to the given duration.

The codes for the set of reacher moves are $\{$\state\ $s$ \codechoose\ $d\ \mid\ \exists s', (s,d,s') \in T \}$,
which correspond to the choice by player~$1$ of an available duration;
$\{$\finish$\}$, played when the winning configuration is reached;
$\{$\cancel\ $(d,s')$ \erase\ $(j,a)\ \mid\ (s_d,d,s') \in T,\ 0 \le j < k,\ 0 \le a \le 3\}$,
to modify the third and fourth part of the counter value and eventually reach~$0$ when it seems that the opponent cheated;
and $\{$\cancel\ $(d,s')$ \remove\ $d'\ \mid\ (s_d,d,s') \in T,\ d \not= d'\}$,
to point out a cheating from the opponent, i.e., cancel the last move and subtract the real duration that was chosen.
When we do not specify the parameters like $s$ and $d$ in the codes, we write the type of the moves,
for example \state/\codechoose.

We call \kw{good encoding} a sequence that alternates reacher and opponent moves such that:
\begin{itemize}
\item The first move is \state\ $s_0$ \codechoose\ $d$ for a certain $d$.
\item The last move is \finish\ and the move before is \duration\ $d$ \goto\ $\bot$ for a certain $d$.
\item There are neither \cancel/\erase\ nor \cancel/\remove\ nor other \finish\ moves.
\item For two consecutive moves \state\ $s$ \codechoose\ $d$ and \duration\ $d'$ \goto\ $s'$, we have $d$ = $d'$.
\item For two consecutive moves \duration\ $d$ \goto\ $s$ and \state\ $s'$ \codechoose\ $d'$, we have $s$ = $s'$.
\end{itemize}
All sequences that are not the prefix of a good encoding and that have no good encoding as a prefix are \kw{bad encodings},
and the first move that refutes in this case the first or one of the three last properties of a good encoding is called \kw{deviating move}.
A sequence that only refutes the second property is neither a good nor a bad encoding,
and we deal separately with sequences that continue after a finish move.

Note that for each play in the restricted countdown game the players have the possibility in the robot game
to build with their moves a good encoding and the codes of this good encoding trace the play.
The hard part is to handle bad encodings.
To understand how it can be done, let us give the integers that correspond to each code.

\begin{itemize}
\item \duration\ $d_i$ \goto\ $s'$ == $-4^{i} + 4^{h+s'}$;
\item \state\ $s$ \codechoose\ $d_i$ == $4^{i} - 4^{h+s} - d_i \cdot 4^{h+n}$;
\item \finish\ == $-4^{h+n-1} - k \cdot 4^{h+n+k} - 4^{h+n+k+k'}$;
\item \cancel\ $(d_i,s')$ \erase\ $(j,a)$ == $-$(\duration\ $d_i$ \goto\ $s'$)$ - a \cdot 4^{h+n+j} - 4^{h+n+k}$;
\item \cancel\ $(d_i,s')$ \remove\ $d_j$ == $-$(\duration\ $d_i$ \goto\ $s'$)$ - 4^{j} - 4^{h+n+k+k'}$.
\end{itemize}

It appears that every opponent move is positive and every reacher move is negative.
Moreover, any opponent move plus any reacher move is negative,
hence, if the counter value becomes negative, the opponent wins.

The next proposition shows the need for both players to build good encodings.

\label{prop7}\begin{proposition}
If a sequence is a bad encoding, then the adversary of the player who has played the deviating move
has a winning strategy from this move onwards.
\end{proposition}

\begin{proof}
Let us consider a bad encoding and every possibility for the deviating move:
\begin{itemize}
\item An opponent move \duration\ $d_i$ \goto\ $s'$ whereas the expected duration was $d_j$.

In this case, the counter value has the $i^{th}$ digit at $3$.
The reacher then has the occasion to play
\cancel\ $(d_i,s')$ \remove\ $d_j$, in order that the first two parts of the counter value become $0$.
From this point on, the reacher just has to cancel every further opponent move
and erase step by step every digit of the third part of the counter until the value $0$ is encountered.

This case is illustrated in the Figure~$4$. The first line corresponds to the second line in the Figure~$3$.
Imagine that the opponent plays the deviating move \duration\ $6$ \goto\ $0$ (second line).
The reacher can react with \cancel\ ($6$,$0$) \remove\ $3$ (third line), and then, whatever the opponent does
(fourth and sixth line), the reacher cancels every move and erases the first and second digits of the third part
(fifth and seventh line), which makes him win because the counter is~$0$.

\item A reacher move \state\ $s$ \codechoose\ $d_i$ whereas the expected location was $s'$.

In this case, the counter value has the $h+s^{th}$ digit at $3$.
Now, the opponent can take advantage of this error and play a move with \goto\ $s$ if, and only if,
the $h+s^{th}$ digit has been lowered to $2$ by the previous reacher move,
therefore this digit will always be $3$ after an opponent move and never $0$ again after a reacher move,
because none of them permits to increase a digit of the second part or to decrease it by $2$ or more,
even with carries. In particular, the counter value cannot become $0$.

\item A reacher \cancel/\remove\ move.

Here, the first part of the counter value had only digits at $0$ just before
because the opponent did not do a deviating move and now one digit is at $3$, let us say it is the $i^{th}$ one.
The opponent will always use moves with \duration\ $d_i$ when the $i^{th}$ digit is $3$ and other moves when it is $2$
such that this digit can no longer be put to $0$ after the reacher plays, therefore the opponent wins.

\item A reacher \cancel/\erase\ move.

Here, the fourth part of the counter value has been reduced, hence the move \finish,
which would lead to a negative counter value, should be avoided by the reacher.
In other words, the opponent just has to match the duration of his move to the one that the reacher chose right before,
like in a good encoding, to be sure that he wins.
Indeed, the only possibility for the reacher to win is now to use a \cancel/\remove\ move,
but he will lose whenever he does this meanwhile the first part of the counter value has only digits at $0$.

\end{itemize}
\end{proof}

\begin{figure}
\begin{center}
\begin{tabular}{| c | c | c | c || l |}
\hline
$0\ 0\ 1\ 0\ 0\ 0$ & $0\ 0\ 0\ 0$ & $1\ 1$ & $2\ 1$ & Expected duration: $3$\\
\hline

$0\ 0\ 1\ 0\ 0\ 3$ & $0\ 0\ 0\ 0$ & $1\ 1$ & $2\ 1$ & \duration\ $6$ \goto\ $0$\\
\hline

$0\ 0\ 0\ 0\ 0\ 0$ & $0\ 0\ 0\ 0$ & $1\ 1$ & $2\ 0$ & \cancel\ $(6,0)$ \remove\ $3$\\
\hline

$0\ 0\ 0\ 0\ 3\ 3$ & $3\ 3\ 3\ 0$ & $1\ 1$ & $2\ 0$ & \duration\ $5$ \goto\ $\bot$\\
\hline

$0\ 0\ 0\ 0\ 0\ 0$ & $0\ 0\ 0\ 0$ & $0\ 1$ & $1\ 0$ & \cancel\ $(5,\bot)$ \erase\ $(1,1)$\\
\hline

$0\ 0\ 0\ 0\ 3\ 3$ & $3\ 3\ 3\ 0$ & $0\ 1$ & $1\ 0$ & \duration\ $5$ \goto\ $\bot$\\
\hline

$0\ 0\ 0\ 0\ 0\ 0$ & $0\ 0\ 0\ 0$ & $0\ 0$ & $0\ 0$ & \cancel\ $(5,\bot)$ \erase\ $(2,1)$\\
\hline
\end{tabular}
\caption{Deviating move of the opponent and reaction of the reacher.}
\end{center}
\end{figure}

We now restrict to good encodings, for which the next proposition decides the winner
depending on the winner of the corresponding play in the countdown game. We first need the following lemma.

\label{lemma2}\begin{lemma}
Consider a prefix of a good encoding without any \finish\ move.
The following invariants hold:
\begin{itemize}
\item The digits in the first part of the counter value are all~$0$ after an opponent move and all~$0$ except one~$1$ after a reacher move.
\item The digits in the second part of the counter value are all~$0$ after a reacher move and all~$0$ except one~$1$ after an opponent move.
\end{itemize}
\end{lemma}

\begin{proof}
At the beginning of the play, before the first reacher move, there is one~$1$ and other digits are~$0$ in the second part,
and all the digits in the first part are~$0$. This can also be seen in the Figure~$3$.
Consider a good encoding. The following alternation happens:
reacher \state/\codechoose\ moves erase the~$1$ in the second part and increment a digit, hence a~$0$, in the first part;
and opponent moves erase the~$1$ that appeared in the first part and increment a~$0$, in the second part.
\end{proof}

\label{prop8}\begin{proposition}
Consider a good encoding $\pi'$ built from a play $\pi$ in the restricted countdown game.
If player~$1$ wins $\pi$, then the reacher wins any play 
that begins with the prefix $\pi'$ in the robot game when he plays the \finish\ move,
else the opponent has a winning strategy after the \finish\ move.
\end{proposition}

\begin{proof}
We look at the evolution of the counter value along a good encoding.
Every \duration\ $d$ cancels the previous \codechoose\ $d$,
every \state\ $s$ cancels the previous \goto\ $s$, \state\ $s_0$ cancels the $1$ in the initial counter value,
and \finish\ cancels the digit that correspond to $\bot$ and the fourth part.
In other words, all parts except possibly the third one are zero at the end of a good encoding.
Here, we do not consider possible carries from the third to the fourth part, which make the reacher lose.
As for the third part, it first represents the initial value in the countdown game and the reacher subtracts from it
the values of the durations chosen by player~$1$ in the simulated play.
The counter value is also $0$ at the end of the good encoding in the robot game if, and only if,
the corresponding play in the countdown game ends in the configuration $(\bot,0)$.

Let us show why we need the \finish\ move. According to Lemma~$18$,
the reacher cannot win if he plays \state/\codechoose\ moves forever,
even if he tries to make a carry appear from the third to the fourth part of the counter value,
because there will always be a digit at~$1$ in the first part of the counter.
However, with a \finish\ move, no digit is incremented in the first part of the counter.
Hence, the reacher should use this move at least once at the end of a good encoding.

Note that, in particular, the reacher loses if he subtracts to the third part more than the initial value in the countdown game.

Now, we present the winning strategy for the opponent if the reacher did not win at the moment where he played \finish.
The fourth part of the counter is now nullified,
hence the reacher can afterwards only use \state/\codechoose\ moves because other moves would make the counter value negative.
Consequently, the opponent can do a move with \goto\ $\bot$
and guarantee at the next step that the digit that corresponds to $\bot$ is never $0$ again.
\end{proof}

We conclude from Propositions~$17$ and $19$ that the reacher has a winning strategy in the robot game if, 
and only if, player~$1$ has a winning strategy in the countdown game.
Indeed, both players need to generate a good encoding, else they know that they will lose,
thus, it is just a matter of checking whether the reacher can enforce the good encoding to make him win.

\section{Conclusion and perspectives}

In this paper, we give an EXPTIME algorithm for solving robot games on the integer line,
and prove EXPTIME-hardness by a reduction from countdown games.
According to \cite{DR13}, it is open whether deciding the winner of robot games is decidable in dimension two.
It will be interesting to see if the game on a grid has enough regularity properties so as to adapt our algorithm for the dimension two problem.

\bibliographystyle{eptcs}
\bibliography{robot_games_full}

\end{document}